\documentclass[prl,aps,twocolumn,notitlepage,superscriptaddress,showpacs,nofootinbib]{revtex4-2}

\usepackage{bm,braket}
\usepackage{bbold,dsfont}
\usepackage{units}
\usepackage{amssymb}
\usepackage{amsthm}
\usepackage{mathtools}
\usepackage{mathrsfs}

\usepackage{booktabs}
\usepackage[table]{xcolor}
\definecolor{lightgray}{gray}{0.95}

\usepackage{graphicx}
\usepackage[colorlinks]{hyperref}
\hypersetup{
	colorlinks = true,
	urlcolor = {blue},
	citecolor = {blue},
	linkcolor= {blue}
}

\newtheorem{theorem}{Theorem}
\newtheorem{corollary}[theorem]{Corollary}

\newtheorem{definition}{Definition}

\def\>{\rangle} 
\def\<{\langle}

\DeclareMathOperator{\tr}{tr}

\begin{document} 
 
\title{Robust projective measurements through measuring code-inspired observables}

\author{Yingkai Ouyang}
\email{y.ouyang@sheffield.ac.uk}
\affiliation{Department of Physics and Astronomy, University of Sheffield, Sheffield, S3 7RH, United Kingdom}

\begin{abstract} 
Quantum measurements are ubiquitous in quantum information processing tasks, but errors can render their outputs unreliable.
Here, we present a scheme that implements a robust projective measurement through measuring code-inspired observables.
Namely, given a projective POVM, a classical code and a constraint on the number of measurement outcomes each observable can have, 
we construct commuting observables whose measurement is equivalent to the projective measurement in the noiseless setting.
Moreover, we can correct $t$ errors on the classical outcomes of the observables' measurement if the classical code corrects $t$ errors.
Since our scheme does not require the encoding of quantum data onto a quantum error correction code,
it can help construct robust measurements for near-term quantum algorithms that do not use quantum error correction.
Moreover, our scheme works for any projective POVM, and hence can allow robust syndrome extraction procedures in non-stabilizer quantum error correction codes. 
\end{abstract}

\maketitle

\section{Introduction}
Quantum measurements, ubiquitous in quantum information processing tasks, are basic building blocks used in all quantum algorithms, such as in quantum sampling \cite{tillmann2013experimental,lund2017quantum,wild2021quantum}, quantum learning \cite{learn-unitary2010,arunachalam2017guest,haah2017sample,learnTgates,ouyang-learn-graph-products}, 
quantum channel estimation \cite{escher2011general,Hayashi11,pirandola2019fundamental,zhou2021asymptotic},
quantum parameter estimation \cite{HELSTROM1967101,helstrom,holevo,HM08,Albarelli2019_PRL,sidhu2020tight,CSLA,HO}, or universal quantum computations \cite{mbqc-PhysRevA.68.022312,van2013universal,menicucci2006universal,mbqc-briegel2009measurement}.
However, errors in quantum measurements prevent these quantum algorithms from unlocking their full potential.

Quantum algorithms use either just the classical outputs of quantum measurements or both the classical outputs and the measured states.
%In the measurement of a quantum state, two types of errors can occur. Errors can corrupt not only the quantum state, but also the classical outcome if the measurement.
Near-term quantum algorithms
such as quantum sampling, quantum learning, and quantum parameter estimation algorithms
use primarily the classical outputs of quantum measurements.
When errors afflict the classical outcomes these near-term quantum algorithms' measurements, 
the precision of these quantum algorithms' outputs suffers.
Regarding near-term quantum algorithms, 
there has been a plethora of recent recent on the topic of quantum error mitigation \cite{geller2020rigorous,maciejewski2020mitigation,mitigate-meas-errors-PhysRevA.103.042605,mitigation-meas-errors-PRXQuantum.2.040326,cai2023quantum,metrology-noisy-meas-PRXQuantum.4.040305},
where the goal is to reduce the statistical error of quantum measurements.
This is achieved through repeated experiments and classical post-processing of the additional classical data obtained.
However, the question of how to directly correct such measurement errors in these near-term algorithms without access to quantum error correction (QEC) is an open problem.

However, these mitigation schemes do not correct the measurement errors that occur.
Hence arises the question that has been open since the dawn of the research field of quantum computing:
 
Universal quantum computations can use both quantum and classical outputs of measurements. 
Correction of both quantum and classical errors in measurements using stabilizer codes has been discussed in the context of data-syndrome codes \cite{ashikhmin2014robust,fujiwara2014,ashikhmin2016correction,ashikhmin2020quantum,kuo2021decoding,nemec2023quantum,guttentag2023robust}, single-shot QEC \cite{bombin2015single,campbell2019qss,quintavalle-PRXQuantum.2.020340}, and fault-tolerant quantum computing \cite{campbell2017roads}. 
However, the pertinent question of how to correct measurement errors for non-stabilizer codes, such as for bosonic codes \cite{CLY97,GKP01,BinomialCodes2016,ouyang2019permutation,rotation-invariant-PhysRevX.10.011058}, remains unanswered.

Here, we present a scheme that implements a robust projective measurement through measuring code-inspired observables.
Namely, given a projective POVM, a classical code and a constraint on the number of measurement outcomes each observable can have, 
we construct commuting observables whose measurement is equivalent to the noiseless projective measurement.
Moreover, we can correct $t$ errors on the classical outcomes of the observables' measurement if the classical code corrects $t$ errors.
The minimum number of commuting observables required depends on 
(1) the number of measurement outcomes for each commuting observable, 
(2) the number of measurement outcomes for the underlying projective measurement,
and (3)
the number of errors on classical outcomes that we wish to correct.
We obtain bounds on the minimum number of commuting observables required based on bounds on the parameters of classical codes.

We suggest how to implement our scheme using ancillary coherent states. The requirements are modest. 
Namely, we need access to a linear coupling between the observables and ancillas, and the ability to perform homodyne measurement on the ancillas.
Hence, using a modest amount of quantum control, we can in fact correct measurement errors, without need for quantum error correction codes.

We explain how our scheme allows the correction of measurement errors in any QEC code that satisfies the Knill-Laflamme QEC criterion \cite{KnL97}. 
Namely,
given any QEC code that corrects a set of errors $\mathfrak K$,
we bound the minimum number of commuting observables $n_{\mathfrak K,t}$ required to 
correctly perform the syndrome extraction stage in the Knill-Laflamme recovery procedure if there are up to $t$ errors on the syndrome.
Based on this, we give bounds on $n_{\mathfrak K,t}$, and elucidate this bound for binary QEC codes and the binomial code.

We envision our scheme to complement existing quantum error mitigation techniques, and thereby enhance the performance of near-term quantum algorithms. 
In the longer term, our scheme can also enhance the design of fault-tolerant quantum computations on non-stabilizer codes, such as those reliant on bosonic codes \cite{rotation-invariant-PhysRevX.10.011058,grimsmo2020quantum,noh2020fault}.

\begin{figure}
    \centering
    \includegraphics[width = 0.49\textwidth]{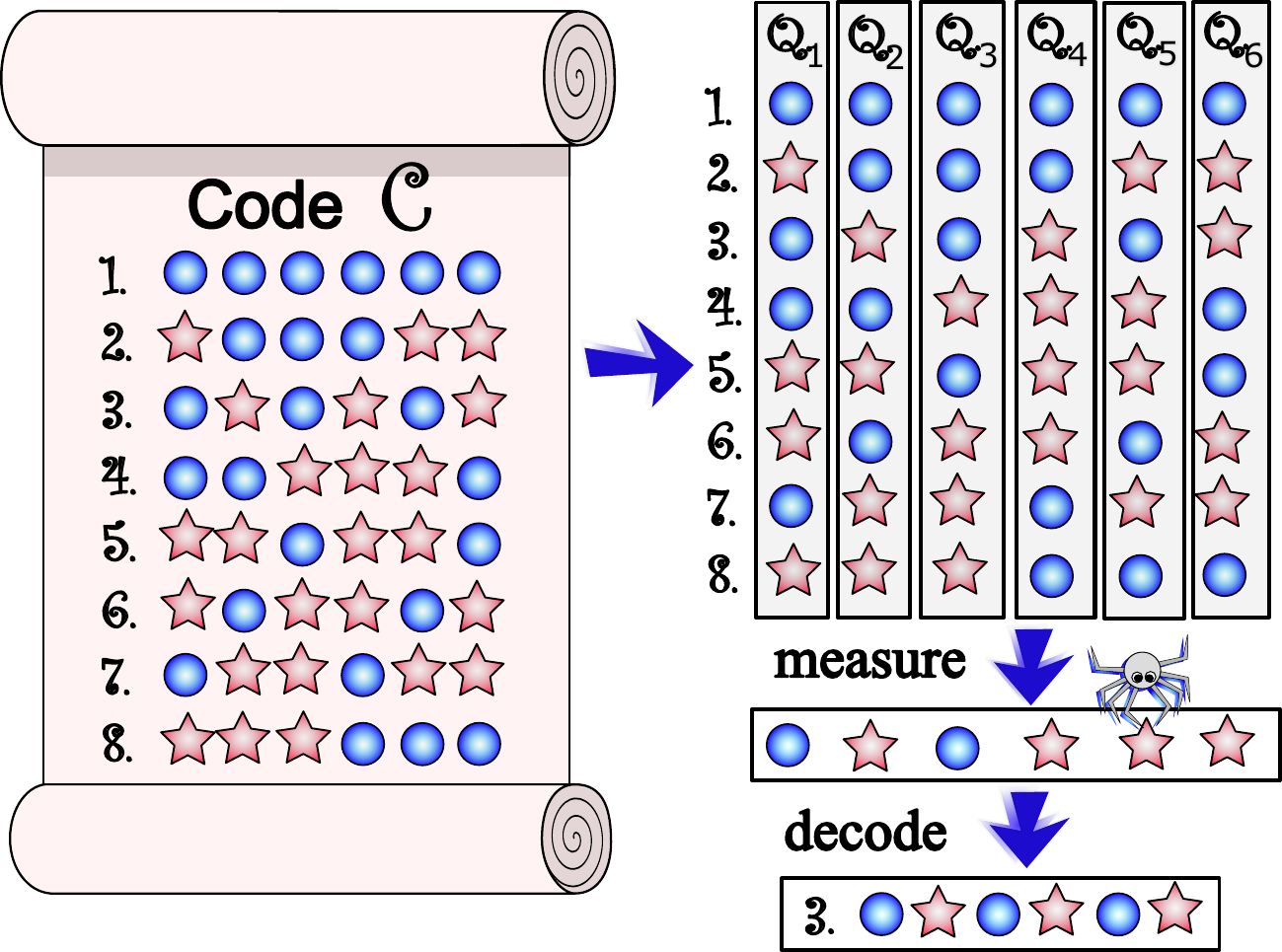}
    \caption{{\bf Our scheme.} Suppose that a projective measurement $P$ projects a quantum state into one of 8 orthogonal subspaces. We label each subspace with a codeword of a classical code $C$. We depict a shortened Hamming code with 8 codewords. Each codeword is a 6-bit string, and the code corrects one error. Then we illustrate the commuting $q$-observables $Q_1,\dots, Q_6$ as columns on the right side of the diagram. Each $q$-observable is an appropriate linear combination of projectors in $P$. Here, $q=2$, corresponding to binary outcome observables. 
    A spider illustrates an error that occurs on the measurement outcome of $Q_5$.
    We can correct this error using the decoder of $C$, recovering the correct measurement outcome. We conclude that the quantum state has been projected to the third subspace.
    }
    \label{fig:scheme}
\end{figure}

\section{Measurements}

We can describe a measurement as a POVM \cite{NielsenChuang}, which is a set of positive operators that sum to the identity operator. 
Without loss of generality, we can always focus on projective POVMs, where the positive operators are furthermore pairwise orthogonal projectors. 
This is because Naimark's theorem ensures that 
for any POVM, 
we can always perform a projective POVM on an extended Hilbert space \cite{beneduci2020notes}.

From the Born rule, 
measuring a projective POVM $P \coloneqq \{P_1, \dots, P_M\}$ 
with pairwise orthogonal projectors  
on an input state $\rho$
yields the {\em post-measurement state}
$  \rho_k \coloneqq  P_k \rho P_k / \tr[\rho P_j]$
with probability
$p_k \coloneqq \tr[\rho P_k]$.
We denote the measurement's output as $(\rho_k,k)$ where $k$ is the  measurement's {\em classical outcome} that allows us to uniquely identify the post-measurement state $\rho_k$. 

Mathematically, an observable is a Hermitian operator.
Consider an observable $O = \sum_k \lambda_k P_k$, where $\lambda_k$ are distinct real numbers for different values of $k$.
Measurement of $O$ on $\rho$ gives an output $(\rho_k, \lambda_k)$ comprising of a post-measurement state and some eigenvalue of $O$. According to the Born rule, we obtain $(\rho_k, \lambda_k)$ with probability $p_k$.
Since there exists a function that maps $\lambda_k$ back to $k$, the measurement of $O$ is the same as the measurement of $P$.
% Measuring an observable $Q_P \coloneqq \sum_{k=1}^M k P_k$ on $\rho$ is equivalent to measuring $P$ on $\rho$, because the measurement's output would always be of the form $(\rho_k,k)$.

Errors affect a measurement's output in two different ways. First, errors can corrupt the classical outcome $k$.
Such errors can lead us to mistakenly conclude that the post-measurement state is $ \rho_v $ for $v \neq k$ when the true post-measurement state is in fact $\rho_k$.
Second, errors can corrupt the post-measurement state $\rho_k$. 
Here, we propose a measurement scheme that allows correction of errors on classical outcomes.

% The question of how to correct both quantum and classical type of errors in the measurement process is often discussed in the context of
% .This, however, is not the main focus of our paper.

Now, let $q$ be an integer where $q\ge 2$, and let us define a Hermitian operator with $q$ distinct eigenvalues as a $q$-observable. 
Operationally, the integer $q$ counts the number of possible measurement outcomes of each observable. 
% For example, according to our definition, $Q_P$ is an $M$-observable.

\section{Commuting observables from classical codes}

In the observable $ O $, the integers $ 1, \dots, M $ label $ M $ distinct measurement outcomes. 
Consider a classical code $ C $ comprising of $ M $ distinct codewords.
When $C$ is a $q$-ary code of length $n$, 
each codeword is a vector in 
$ \{0,1,\dots, q-1\}^n $.
We denote 
\begin{align}
{\bf x}^{(k)} = (x^{(k)}_{1},\dots,x^{(k)}_{n})  \notag 
\end{align}
as the $k$th codeword of $C$,
and we can write $C=\{{\bf x}^{(k)}:k=1,\dots,M\}$.
% If $ q = 2 $, $ C $ corresponds to an $ n $-bit binary code; otherwise, $ C $ is a non-binary $q$-ary code.

Each integer $1,\dots, M$ labels exactly one codeword in $C$. 
The encoder $E_C$ of $C$ is a bijective map from the classical labels in $\{1,\dots,M\}$ to codewords in $C$. 
Namely, $E_C(k) = {\bf x}^{(k)}$. 
Without errors on the components of ${\bf x}^{(k)}$, a decoder of $C$ performs the inverse map of $E_C$, 
and maps the codeword ${\bf x}^{(k)}$ back to the label $k$.

In the measurement of $P$, errors could afflict its classical outcome. 
To address this, we propose the measurement of $n$ commuting $q$-observables $Q_1,\dots, Q_n$ that encode redundant information about $P$. 
We denote the classical outcome of  $Q_j$'s measurement as $y_j$ and denote the output of the measurements of $Q_1,\dots, Q_n$ as 
$(\tau,{\bf y})$ 
where $\tau$ denotes the post-measurement state and ${\bf y} = (y_1,\dots, y_n)$.
We want the $q$-observables to be {\em consistent} with $P$, in the sense that measurement of the $q$-observables performs the same measurement as $P$ in the noiseless setting.
Hence we give the following definition.
\begin{definition}
Let $P$ be a projective POVM and $Q_1,\dots, Q_n$ be commuting observables. 
The observables $Q_1,\dots, Q_n$ are consistent with $P$ if 
there exists a function $f$ such that for any output $(\tau,{\bf y})$ of the measurement of $Q_1,\dots, Q_n$ on $\rho$,
we have $\tau = \rho_{f({\bf y})}$.
\end{definition}

We propose to construct $q$-observables using information about a projective POVM $P$ and a classical $q$-ary code $C$.
Namely, for $j=1,\dots, n$ we define the $q$-observables as
\begin{align}
    Q_j(C,P) \coloneqq \sum_{k=1}^M x^{(k)}_{j} P_k.\label{def:qj}
\end{align}
When the context is clear, we use $Q_j$ to denote $ Q_j(C,P)$.
From the orthogonality of the projectors $P_k$, the observables 
$Q_1,\dots, Q_n$ are pairwise commuting, which allows us to measure $Q_1,\dots, Q_n$ in any order.

In our construction,
the correctibility of errors on the measurement outcomes of our $q$-observables depends on the minimum distance of $C$, given by
\begin{align}
    d(C) \coloneqq \min_{{\bf y} \neq  {\bf z} \in C} d_H({\bf y} , {\bf z}),\notag
\end{align}
where
$   d_H({\bf y},{\bf z}) \coloneqq |\{i : y_i \neq z_i\}|$
 is the Hamming distance between 
 tuples
 ${\bf y}$ and ${\bf z}$.
Namely, we can correct any $t(C) \coloneqq \lfloor (d(C)-1)/2 \rfloor$ errors on the classical outcomes of $Q_1, \dots, Q_n$.

A decoder of a classical code can correct up to $t(C)$ measurement errors on ${\bf y}$.
This is because a noiseless ${\bf y}$ must be a codeword in $C$.
Here, a decoder $\mathcal D$ of a code $C$ is a function $\mathcal D : \{0,\dots,q-1\}^n \to \{1,\dots,M\}$
which maps an $n$-tuple to an index that labels the codewords.
Given some non-negative integer $a$, we say that $\mathcal D$ is an $a$-decoder of $C$, if for all $k=1,\dots, M,$ and for all ${\bf y}$ such that 
$d({\bf y},{\bf x}_k)\le a$, we have 
\begin{align}
    \mathcal D({\bf y}) = k.
\end{align}
An $a$-decoder corrects $a$ errors.
When $d(C)=d$, then there is a 
$t(C)$-decoder for $C$.
Our main result is the following.
\begin{theorem}
\label{thm:main}
Let $C$ be a $q$-ary code of length $n$, and 
let $\mathcal D$ be a 
$t(C)$-decoder for $C$.
We measure $Q_1(C,P), \dots, Q_n(C,P)$ on a quantum state $\rho$, and obtain the classical outcome ${\bf y}=(y_1, \dots, y_n)$ along with the post-measurement state $\tau$.
Suppose that at most $t(C)$ components of ${\bf y}$ have been corrupted.
Then $\tau = \rho_{\mathcal D({\bf y})}$.
\end{theorem}
\begin{proof}
Let ${\bf z}=(z_1,\dots,z_n)$ denote the classical outcome if no errors occurred.
Then, ${\bf z}-{\bf y}$ has a Hamming weight of at most $t$.
Furthermore, we have 
\begin{align}
    \mathcal D({\bf z}) = 
    \mathcal D({\bf y}) .\label{correct decoding}
\end{align}
 % By the definition of 
 % $Q_j(C,\Pi)$ in \eqref{def:qj}, for all $k = 1,\dots, M,$
 % each $P_k$ is labeled by ${\bf x}_k \in (x_{k,1},\dots,x_{k,n})$, and by definition,
 % ${\bf x}_k $ is a codeword in $C$. 
\noindent
{\bf Case 1: No errors on measurement outcomes.}
When we measure the observable
 $Q_j(C,P)$ and obtain the classical outcome $z_j$, the resultant state must be on the support of the projector 
\begin{align}
    P_{j,z_j} =  \sum_{k: z_j = x_{k,j}} P_k.
\end{align}
After measuring the observables
$Q_1(C,P), \dots, Q_n(C,P)$, 
we obtain the classical outcomes
$z_1, \dots, z_n$.
Then the state $\tau$ is on the support of
\begin{align}
    \prod_{j=1}^n P_{j,z_j}
    = 
    \sum_{k: z_j = x_{k,j}, j=1,\dots,n} P_k 
    =
    \sum_{k: {\bf z} = {\bf x}_k} P_k.
    \label{eq:prod projs}
\end{align}
From \eqref{eq:prod projs}, ${\bf z}$ must belong to $C$.
Since there are no repeated codewords in $C$, there is a unique $k$ for which ${\bf z} = {\bf x}_k$.
Together with the fact that $t(C) \ge 0$, it follows that
\begin{align}
    \prod_{j=1}^n P_{j,z_j}
    = 
P_{\mathcal D({\bf z})}.
    \label{eq:prod projs2}
\end{align}
The state $\tau$ must be on the support of $P_{\mathcal D({\bf z})}$,
which means that $\tau = \rho_{\mathcal D({\bf z})}$.
\newline

\noindent
{\bf Case 2: At most $t(C)$ errors on classical outcomes.} 
From case 1, we know that $\tau = \rho_{\mathcal D({\bf z})}$.
Since $\mathcal D({\bf y}) = \mathcal D({\bf z})$, we have $\tau = \rho_{\mathcal D({\bf y})}$. 
\end{proof}
In our proof of Theorem 1, we show that in the noiseless setting, the $n$-tuple of classical outcomes is a codeword of $C$.
When there are at most $t(C)$ errors on the classical outcomes, 
the decoder $\mathcal D$ corrects these errors.
Hence, the observables $Q_1(C,P),\dots, Q_n(C,P)$ are consistent with $P$, even in the presence of some errors on the classical outcomes.

As an example, consider a scheme that uses the shortened Hamming code $C_6$ and a projective POVM $P = \{P_1,\dots , P_8\}$ to define the six binary observables to measure both in the noiseless setting.
In this example, the parameters of the code are $q=2,n=6,d(C_6)=3,$ and $M=8.$
Now the code $C_6$ is a linear code generated by binary vectors $a_1 = \texttt{100011}$,
$a_2 = \texttt{010101},$
and $a_3 = \texttt{001110}$, and has eight codewords given by
\begin{align}
{\bf x}_{1} &=  \texttt{000000},\notag\\
{\bf x}_{2} &= \texttt{100011} =a_1,\notag\\
{\bf x}_{3} &= \texttt{010101} =a_2,\notag\\
{\bf x}_{4} &= \texttt{001110} = a_3,\notag\\
{\bf x}_{5} &= \texttt{110110} = a_1+a_2  ,\notag\\
{\bf x}_{6} &= \texttt{101101} =a_1+a_3 ,\notag\\
{\bf x}_{7} &= \texttt{011011} =a_2+a_3 ,\notag\\
{\bf x}_{8} &= \texttt{111000} =a_1+a_2 +a_3.\label{C6 codewords}
\end{align}
Applying the definition \eqref{def:qj} along with the form for the codewords in \eqref{C6 codewords}, 
the corresponding binary observables are
\begin{align*}
    Q_1(C_6,P) &= P_2+P_5+P_6+P_8\\
    Q_2(C_6,P) &= P_3+P_5+P_7+P_8\\
    Q_3(C_6,P) &= P_4+P_6+P_7+P_8\\
    Q_4(C_6,P) &= P_3+P_4+P_5+P_6\\
    Q_5(C_6,P) &= P_2+P_4+P_5+P_7\\
    Q_6(C_6,P) &= P_2+P_3+P_6+P_7.
\end{align*}
We illustrate these binary observables in Figure \ref{fig:scheme}.
Now consider no errors on classical outcomes.
When we measure $Q_1(C_6,P)$ and obtain the classical outcome 0,
the state must be on the support of $I-Q_1(C_6,P)=P_1+P_3+P_4+P_7$, where $I$ denotes the identity operator.
If we measure $Q_2(C_6,P)$ and obtain the classical outcome 1, then the state is on the support of $P_1+P_3+P_4+P_7$ and $P_3 + P_5 +P_7+P_8$. Hence the state is on the support of $P_3+P_7$.
If we measure $Q_3(C_6,P)$ and obtain the classical outcome 1, then the state is on the support of $P_3+P_7$ and $P_4+P_6+P_7+P_8$.
Hence the state is on the support of $P_7$. Further measurements of the observables $Q_4(C_6,P), Q_5(C_6,P), Q_6(C_6,P)$ give redundant information about where the state is projected on, and 
we obtain the codeword ${\bf x}_{7}$ as the classical outcome. 

In Figure \ref{fig:scheme} we illustrate the measurement of the binary observables $Q_1(C_6,P),\dots, Q_n(C_6,P)$ when an error afflicts the classical outcome of $Q_5(C_6,P)$.

\section{Implications}

\noindent
{\bf Combinatorics:-}
 {\em What is the minimum number of $q$-observables required to correct $t$ errors on the classical outcome of a projective POVM with $M$ projectors}? We answer this question in the following.
\begin{corollary}
\label{coro}
Let $P$ be a projective POVM with $M$ projectors.
Let $n_q(M,d)$ to be the shortest $n$ such that there exists a code of length $n$ and with at least $M$ codewords and distance at least $d$.
Let $n_{q,t,M}$ be the smallest integer
such that there exist observables 
$Q_1,\dots, Q_n$ consistent with $P$, even after any $t$ errors occur on the classical outcomes of $Q_1,\dots, Q_n$.
Then $n_{q,t,M}=n_q(M,2t+1)$.    
\end{corollary}
\begin{proof}
From Theorem \ref{thm:main}, we know that 
the condition for $Q_1,\dots, Q_n$ to be consistent with $P$ after $t$ errors occur on the classical outcomes
is equivalent to the condition that a $q$-ary classical code $C$ has length $n$, distance at least $2t+1$, and has $M$ codewords.
\end{proof}
The combinatorics of $n_q(M,d)$ directly relates to the combinatorics of $A_q(n,d)$, where $A_q(n,d)$ is the maximum number of codewords in a $q$-ary code with Hamming distance $d$ and with codewords having $n$ components. 
Note that $n_{q,t,q}=2t+1$ through the use of a $q$-ary repetition code.
Using results on the combinatorics of $A_q(n,d)$ and $n_q(M,d)$ \cite{bounds25}, we illustrate the values of $n_{q,t,M}$,
in Table \ref{table} for $q=2$, $t=1,2,3$ and $2\le M \le 40$.

\begin{table}[ht]
  \centering
  \setlength{\tabcolsep}{7.5pt} % Adjust the horizontal spacing
  \begin{tabular}{c|*{8}{c}}
    \toprule
    % \rowcolor{lightgray}
    \textbf{$t \backslash M$} & \textbf{2} & \textbf{4} & \textbf{6} & \textbf{8} & \textbf{12} & \textbf{16} & \textbf{20}
    & \textbf{38-40}\\
    \midrule
    % \rowcolor{lightgray!90}
    \textbf{1} & 3 & $  6$ & $  7$ & $ 7$  & $  8$  & $ 8$  & $9$ & 10  \\
    % \rowcolor{lightgray!80}
    \textbf{2} & 5& $9$  & $ 10$  & $11$  & $11$  & $12$  & $12$ & 14  \\
    % \rowcolor{lightgray!70}
    \textbf{3} & 7 & $12$   & $14$   & $14$  & $15$  & $15$  & $16$ & 18\\
    \bottomrule
  \end{tabular}
    \caption{Some values of $n_{2,t,M}$.}
  \label{table}
\end{table}

When the number $M$ of projectors in $P$ is very large, 
we can bound $M$ in terms of the volume of a $q$-ary Hamming ball of radius $t$, which we denote as $V_{q,n}(t) \coloneqq \sum_{j=0}^t \binom n j (q-1)^j$.
Namely,
%Based on asymptotic bounds for $n_q(M,d)$ (for large $M$), we get
\begin{align}
q^n / V_{q,n}(2t) \le M \le q^n / V_{q,n}(t),
\label{classical-bounds}
\end{align}
where the upper and lower bounds are the
Hamming bound and Gilbert-Varshamov bound respectively 
\cite{sloane}.
bounds such as Johnson's bound \cite{DBLP:journals/tit/Johnson62} or linear programming bounds for classical codes \cite{DBLP:conf/focs/NavonS05,DBLP:journals/amco/MounitsEL07} can tighten the upper bound in \eqref{classical-bounds}.
For large $n$ and $t < n(q-1)/q$, we have $\frac{1}{n}\log_q V_{q,n}(t) = H_q(t/n) + o(1)$,
where $H_q(x) \coloneqq -x \log_q x - (1-x) \log_q (1-x) + x \log_q(q-1)$ denotes the $q$-ary entropy function.
% Now denote $H_q^{-1} : \{0,1\}\to\{0,1-1/q\}$ as the inverse $q$-ary entropy function.
% Asymptotically we have
% \begin{align}
%     1- H_q(2t/n) \le \frac{1}{n}\log_q M \le 1-H_q(t/n).
% \end{align} 
Given $\tau$ as the fraction of errors on the classical outcomes of the measurement of our $q$-observables, for large $M$, we have 
\begin{align}
   \frac{\log_q M}{1- H_q(\tau) + o(1)} \le n_{q,\tau n ,M} \le 
   \frac{\log_q M }{  1-H_q(2\tau) + o(1)}. \label{n-asy-bounds}
\end{align}

\noindent
{\bf Implementation:-}
Similarly to Refs.~\cite{johnsson2020geometric,ouyang2022quantum}, we can couple our quantum state to $n$ bosonic modes initialized as coherent states $|\alpha_1\>,\dots, |\alpha_n\>$
and measure the modes to implement our scheme. 
Let $\hat n_j$ be the number operator on the $j$th mode,
and suppose that $2\pi |\alpha_j|^2\gg q$.
% and the identity operator to all other modes.
The interaction Hamiltonians
\begin{align}
    W_j =\gamma Q_j  \otimes \hat n_j 
\end{align}
model a dispersive coupling between the quantum system and the ancillary bosonic modes.

Now let $|\phi\>$ be a state for which $Q_j |\phi\> = z_j |\phi\>$ for all $j=1,\dots,n$.
Then $W_j|\phi\>|\alpha_j\> = |\phi\> (\gamma z_j \hat n_j |\alpha_j\>)$.
Hence $e^{-i W \theta}|\phi\>|\alpha_j\> 
= |\phi\>e^{-i \theta \gamma z_j \hat n_j}|\alpha_j\>
=|\phi\> | e^{-i \theta \gamma z_j} \alpha_j\>$.
With $\theta = 2\pi/(q \gamma)$,
 the initial phase space distribution of the $j$th mode with radius $|\alpha_j|^2$ and standard deviation $1/\sqrt{2}$ maps to up to $q$ different equiangular rotations in the complex plane.
 Using balanced homodyne detection \cite{scully_zubairy_1997} we can measure the quadratures of the output bosonic fields.
 Because we chose $2\pi |\alpha_j|^2\gg q$, the distributions for different $z_j$ will be distinguishable.
 Hence we project onto the eigenspaces of $Q_j$ in a non-destructive way. 
 Repeating the procedure for $j=1,\dots,n$ allows us to obtain the classical outcome $(z_1,\dots, z_n)$ in the noiseless setting. 
 From Theorem \ref{thm:main}, we can correct up to $t$ errors on $(z_1,\dots, z_n)$ using a classical decoder.
 \newline

\noindent
{\bf Application (Quantum error correction):-}
We can describe the recovery channel of any QEC code as a two-stage process \cite{KnL97}.
In the first stage, we measure a carefully chosen projective measurement with POVM $\Pi'$.
Upon measuring $\Pi'$, we get a classical outcome and a quantum output.
The classical output labels the subspace that the quantum output resides in.
In the second stage, a unitary operation dependent on the classical outcome brings the quantum output back to the codespace.

The projectors in $\Pi'$ depend on the QEC code and the set of operators $\mathfrak K$ to be corrected.
Since the number of correctible spaces of the code is at most $|\mathfrak K|$,
and at most one projector corresponds to an uncorrectible space,
we have $|\Pi'|\le |\mathfrak K|+1$.
For a distance $p$-ary QEC code on $m$ qudits that corrects $k$ errors, we can choose $\mathfrak K$ so that $|\mathfrak K|=V_{p^2,m}(k)$.
From \cite{KnL97}, $|\Pi'|\le |\mathfrak K|$.
Hence,
for an $m$ qubit QEC code that corrects a single error (has distance 3), we have $|\Pi'|\le 2+3m$.

As an example, consider the optimal non-additive nine-qubit binary QEC code that has codespace of dimension 12, and with distance 3 \cite{nonadditive-PhysRevLett.101.090501}. 
In this case $|\Pi'|\le 29$
From Table \ref{table}, 
deploying our scheme with 10 binary observables allows the correction of up to one error on the classical outcome of $\Pi'$. In contrast, the noiseless decoding of this non-additive nine-qubit code in Ref.~\cite{nonadditive-PhysRevLett.101.090501} requires five binary observables, and repeating these measurements thrice to allow the correction of one error necessitates the use of 15 binary observables, which is greater than the 10 binary observables our scheme requires.

Now consider $p$-ary QEC codes that correct $k$ errors using $m$ qudits.
Setting $\tau$ as the maximum fraction of errors on the classical outcome of $\Pi'$, 
from \eqref{n-asy-bounds},
the minimum number $n$ of binary observables required to allow robust syndrome extraction according to $\Pi'$ satisfies the bounds
\begin{align}
    n &\ge \frac{m ( H_{p^2}(k/m) \log_2 p ) + o(1) )}
    {\log_2 H_q(\tau) + o(1)} 
    \\
     n &\le 
    \frac{m ( H_{p^2}(2k/m) \log_2 p ) + o(1) )}
    {\log_2 H_q(2\tau) + o(1)} .
\end{align}

As another example, we consider the binomial code \cite{BinomialCodes2016}, which is a bosonic code on a single mode that corrects gain errors, loss errors and phase errors. 
Here, 
loss errors, gain errors and phase errors are monomials of $a$, $a^\dagger$ and $a^\dagger a$ respectively where $a$ denotes the mode's lowering operator.
Namely, a binomial code that corrects 
$g_1$ gain errors, $g_0$ loss errors, and $k$ phase errors has as its set of correctible errors $\mathfrak K = 
\{a^j : j=0,\dots, g_0\} \cup 
\{(a^\dagger)^j : j=0,\dots, g_1\} 
\cup
\{(a^\dagger a)^j : j=0,\dots, k\}.$
Clearly, $|\mathfrak K|=g_0+g_1+k+1$.
Such a binomial code has two parameters, 
the gap $g = g_0+g_1+1$, and $N = \max\{g0,g1,2k\}$ and encodes one logical qubit,
and is defined by the logical codewords in \cite[Eq.~(7)]{BinomialCodes2016}. 
For such a binomial code where $g_0=g_1=k$, 
we have $|\Pi'| \le 3k+2$.
In Table \ref{table-binom},
we present the minimum number of binary observables that are consistent with $\Pi'$ after the occurrence of up to a single error on the classical outcome of their measurement.

\begin{table}[ht]
  \centering
  \setlength{\tabcolsep}{7.5pt} % Adjust the horizontal spacing
  \begin{tabular}{c|*{8}{c}}
    \toprule
    % \rowcolor{lightgray}
    \textbf{$k$} & \textbf{1} & \textbf{2} & \textbf{3} & \textbf{4} & \textbf{5} & \textbf{6} & \textbf{7}
    & \textbf{8}\\
    \midrule
    % \rowcolor{lightgray!90}
    \textbf{$|\Pi'|$} & 5 & $  8$ & $  11$ & $ 14$  & $  17$  & $ 20$  & $23$ & 26  \\
    % \rowcolor{lightgray!80}
    \textbf{$n$} & 7& $7$  & $8$  & $8$  & $9$  & $9$  & $10$ & 10  \\
    % \rowcolor{lightgray!70}
    \bottomrule
  \end{tabular}
    \caption{Some values of the minimum number of binary observables consistent with $\Pi'$ needed to correct up to one error on the classical outcomes of their measurement $n$, and furthermore correct $k$ gain, loss and phase errors.}
  \label{table-binom}
\end{table}

\section{Discussions}
We proposed a set of commuting $q$-observables whose measurement is consistent with a given projective measurement, even after some errors corrupt the classical outcomes of the measurement of the observables. 
Hence, measuring these commuting observables effectively implements a robust projective measurement.

There is potential to study how near-term quantum algorithms that do not rely on QEC can be improved using our scheme in realistic settings. Moreover, it would be interesting to explore the implementation of our scheme with other non-stabilizer codes, such as concatenated cat codes \cite{CLY97,concat-cats-PRXQuantum.3.010329}, rotation-invariant codes \cite{rotation-invariant-PhysRevX.10.011058}, permutation-invariant codes \cite{Rus00,PoR04,ouyang2014permutation,OUYANG201743,ouyang2019permutation,ouyang2021permutation,aydin2023family}, 
codeword-stabilized codes \cite{CSSZ08}, 
 error-avoiding codes \cite{ZaR97,ouyang2021avoiding,mitigate-coherent-errors},
and certain codes that lie within the ground space of local Hamiltonians \cite{movassagh2020constructing}. 

\section{Acknowledgements}
YO acknowledges support from EPSRC (Grant No. EP/W028115/1).

\appendix
\bibliography{gds}{}
\end{document}